\documentclass[12pt]{amsart}

\usepackage{fullpage}
\usepackage{amsfonts,dsfont,mathtools}
\usepackage{thm-restate}
\usepackage{enumitem}

\usepackage{caption}
\usepackage{xcolor}
\usepackage{tikz,tikz-3dplot}
\usetikzlibrary[calc]

\usepackage[boxruled,linesnumbered,noend]{algorithm2e}

\usepackage[pdftex, colorlinks, unicode]{hyperref}
\usepackage{cleveref}
\usepackage[T1]{fontenc}
\usepackage{microtype}

\newtheorem{thm}{Theorem}[section]
\newtheorem{prop}[thm]{Proposition}
\newtheorem{lem} [thm]{Lemma}
\newtheorem{cor} [thm]{Corollary}
\newtheorem{defi}[thm]{Definition}

\newtheorem{hyp} [thm]{Hypothesis}
\makeatletter
\let\c@algocf\c@thm
\makeatother

\newcommand{\area}{\Pi}

\newcommand{\conf}{\mathcal{C}}
\newcommand{\none}{{\bf None}}
\newcommand{\dual}{\vee}

\newcommand{\ab}{\mathbf{a}}
\newcommand{\bb}{\mathbf{b}}

\newcommand{\nn}{\mathbf{n}}
\newcommand{\uu}{\mathbf{u}}
\newcommand{\vv}{\mathbf{v}}

\newcommand{\xx}{\mathbf{x}}

\DeclareMathOperator{\Span}{Span}
\DeclareMathOperator{\supp}{Supp}

\title{Maximum Overlap Area of Several Convex Polygons Under Translations}

\author{Hyuk Jun Kweon}
\address{Department of Mathematics, University of Georgia, Athens, GA 30602, USA}
\email{kweon@uga.edu}
\urladdr{https://kweon7182.github.io/}

\author{Honglin Zhu}
\address{Department of Mathematics, Massachusetts Institute of Technology, Cambridge, MA 02139, USA}
\email{honglinz@mit.edu}

\date{\today}

\begin{document}

\begin{abstract}
    Let $k \geq 2$ be a constant. Given any $k$ convex polygons in the plane with a total of $n$ vertices, we present an $O(n\log^{2k-3}n)$ time algorithm that finds a translation of each of the polygons such that the area of intersection of the $k$ polygons is maximized. Given one such placement, we also give an $O(n)$ time algorithm which computes the set of all translations of the polygons which achieve this maximum. 
    
\end{abstract}
\maketitle

\section{Introduction}
Shape matching is a critical area in computational geometry, with overlap area or volume often used to measure the similarity between shapes when translated. In this paper, we present a quasilinear time algorithm to solve the problem of maximizing the overlap area of several convex polygons, as stated in the following theorem.

\begin{restatable}{thm}{mainTheorem} \label{thm:main}
  Let $P_0,P_1,\dots,P_{k-1}$ be convex polygons, with a total of $n$ vertices, where $k$ is constant. In $O(n\log^{2k-3}n)$ time, we can finds translations $\vv_0,\vv_1,\dots,\vv_{k-1}$ maximizing the area of
  \[(P_0+\vv_0) \cap \dots \cap (P_{k-1}+\vv_{k-1}).\]
\end{restatable}

Once we have found a placement $\vv_0,\vv_1,\dots,\vv_{k-1}$ that maximizes the overlap area, we can compute the set of all such placements in linear time.
\begin{restatable}{thm}{secondMainTheorem} \label{thm:second main}
  With the notation in \Cref{thm:main}, suppose that we have found a placement $(\vv_0,\vv_1,\dots,\vv_{k-1})$ maximizing the overlap area. Then in $O(n)$ time, we can compute the set of all placements that maximize the overlap area. This set is represented in terms of $O(n)$ linear constraints without redundancy.
\end{restatable}

Suppose that we have $k$ polytopes in $\mathbb{R}^d$ with $n$ vertices in total. Clearly, the overlap volume function under translation is a piecewise polynomial function. To find the maximum overlap volume under translation, we can compute the maximum on each piece. For example, Fukuda and Uno presented an $O(n^4)$ time algorithm for maximizing the overlap area of two polygons in $\mathbb{R}^2$ \cite[Theorem 6.2]{fukuda2007polynomial}. They also gave an $O((kn^{dk+1})^d)$ time algorithm for the problem with $k$ polytopes in $\mathbb{R}^d$ \cite[Theorem 6.4]{fukuda2007polynomial}.

If the polytopes are convex, then the overlap volume function is log-concave. With this additional structure, one may apply a prune-and-search technique and make the algorithm much faster. For example, de Berg et al. gave a highly practical $O(n\log n)$ time algorithm to find the maximum overlap of two convex polygons in $\mathbb{R}^2$ \cite[Theorem 3.8]{de1998computing}. Ahn, Brass and Shin gave a randomized algorithm for finding maximum overlap of two convex polyhedrons in expected time $O(n^3\log^4 n)$ \cite[Theorem 1]{ahn2008}. Ahn, Cheng and Reinbacher \cite[Theorem 2]{ahn2013} find an $O(n\log^{3.5}n)$ time algorithm for the same problem after taking a generic infinitesimal perturbation. The last two results cited from \cite{ahn2008} and \cite{ahn2013} have also been generalized to higher-dimensional cases within the same papers.

On the other hand, there are few known results for problems involving several convex shapes. In this regard, 
the authors 
proposed an $O(n\log^3 n)$ time algorithm to find the maximal overlap area of three convex polygons \cite[Theorem 1.2]{firstPaper}. This result is based on an $O(n\log^2 n)$ time algorithm that finds the maximum overlap area of a convex polyhedron and a convex polygon in $\mathbb{R}^3$ \cite[Theorem 1.1]{firstPaper}. The main algorithm of this paper is a strict generalization of both \cite[Theorem 3.8]{de1998computing} and \cite[Theorem 1.2]{firstPaper}.

The model of computation is the real RAM model. In particular, we assume that in the field of real numbers $\mathbb{R}$, binary operations $+$, $-$, $\times$ and $/$ as well as binary relations $<$ and $=$ can be exactly computed in constant time. We remark that the base field $\mathbb{R}$ can be replaced by any ordered field such as $\mathbb{Q}$ and $\mathbb{R}(\!(\varepsilon)\!)$.

\section{Notation and Terminology}

In this paper, we use the notation $\supp f$ to refer to the closed support of a function $f$, i.e., the closure of the set of points where $f$ is nonzero. Given a set $S$ of vectors over a field $R$, its spanning space is denoted as $\Span_R S$. For two sets $A,B\in\mathbb{R}^d$, we define their Minkowski sum and difference as $A+B=\{\ab+\bb\mid \ab\in A, \bb \in B\}$ and $A-B = \{\xx \mid \xx+B\subset A\}$, respectively. 

We consider closed polytopes unless otherwise specified. When referring to a polytope $P$, its (geometric) interior consists of the set of points not on the facets, while its (geometric) boundary comprises the set of points on the facets. On the other hand, the topological interior of $P\subset{R}^n$ is the set of points in $P$ that have an open ball entirely contained in $P$. The topological boundary of $P$ consists of points that are on the interior of $P$. 
Note that the (geometric) interior is an intrinsic property, while the topological interior is an extrinsic property.

We employ the technique of symbolic infinitesimal translation, similar to \cite{de1998computing}. However, unlike \cite{de1998computing}, our problem requires multiple levels of infinitesimal numbers to handle multiple polygons. Given a field $R$, let $R(\!(\varepsilon)\!)$ be the field of Luarent series of $R$. We work over a very large ordered field
\[\mathbb{R}\langle\!\langle \varepsilon_0,\varepsilon_1,\dots \rangle\!\rangle = \bigcup_{n,s > 0} \mathbb{R}\!\left(\!\left(\varepsilon_0^{1/n},\varepsilon_1^{1/n},\dots,\varepsilon_{s-1}^{1/n}\right)\!\right), \]
which is the field of Puiseux series with countably many variables.\footnote{If the base field $R$ is not $\mathbb{R}$, we may need to take an algebraic closure of $R$.} Here, $\varepsilon_s$ is a positive infinitesimal smaller than any positive expression involving only $\varepsilon_0,\dots,\varepsilon_{s-1}$. Then $\mathbb{R}\langle\!\langle \varepsilon_0,\dots \rangle\!\rangle$ is a real closed field \cite[Theorem~2.91]{basu2006}. Hence, assuming constant time computability for basic operations in $\mathbb{R}\langle\!\langle \varepsilon_0,\dots \rangle\!\rangle$, any algorithm in the real RAM model can be executed with the same time complexity using $\mathbb{R}\langle\!\langle \varepsilon_0,\dots \rangle\!\rangle$.

Of course, $\mathbb{R}\langle\!\langle \varepsilon_0,\dots \rangle\!\rangle$ is far from computable, so we limit our usage of it in this paper.
\begin{defi}
A geometric object in $\mathbb{R}^m$ (such as flats, hyperplanes, polytopes, etc.) is called $\varepsilon(s)$-translated if it is defined by equations and  inequalities that involve only linear polynomials of the form
\[ \ab\cdot\xx + b\]
where $\xx\in\mathbb{R}^m$ is a vector of variables, $\ab\in\mathbb{R}^m$ and $b \in \Span_{\mathbb{R}} \{1,\varepsilon_0,\varepsilon_1,\dots,\varepsilon_{s-1}\}$ are constants.
\end{defi}
By restricting the inputs to $\varepsilon(s)$-translated objects, we can usually ensure that the whole computation is performed within a finite $\mathbb{R}$-vector subspace of $\mathbb{R}\langle\!\langle \varepsilon_0,\dots \rangle\!\rangle$ of dimension $O(1)$. This enables us to apply many algorithms involving $\varepsilon(s)$-translated polytopes with the same time complexity, while guaranteeing the mathematical rigor. Specifically, the following algorithms that we use in our work are valid with $\varepsilon(s)$-translated objects:
\begin{enumerate}
\item Computing intersection of two convex polygons \cite[Section 5.2]{shamos1978computational}
\item Computing intersection of two convex polyhedra \cite{chazelle1993b}
\item Computing maximum sectional area of a convex polyhedron \cite[Theorem 3.2]{avis1996}
\item Computing (1/r)-cuttings \cite{chazelle1993}
\item Solving linear programming \cite{megiddo1984}
\end{enumerate}
Moreover, our algorithm performs computations within
\[ \Span_{\mathbb{R}}\left\{\varepsilon_0^{e_0}\varepsilon_1^{e_1}\dots \varepsilon_{2k-4}^{e_{2k-4}} \;\middle|\; 0\leq e_i\leq 2\right\}.\]

\section{Configuration Space}

The aim of this section is to define the configuration space, the domain of the overlap area function, and discuss its properties. Throughout the paper, we take $k$ convex polygons $P_0, P_1, \dots, P_{k-1}$, where $k$ is a constant. Let $\vv_0,\dots,\vv_{k-1}\in\mathbb{R}^2$ be vectors of indeterminates. The overlap area of
\[I = (P_0+\vv_0) \cap (P_1+\vv_1) \cap \dots \cap (P_{k-1}+\vv_{k-1}) \]
is invariant under the map
\[(\vv_0,\dots,\vv_{k-1})\mapsto (\vv_0+\xx,\dots,\vv_{k-1}+\xx).\]
Therefore, we define the configuration space as a $(2k-2)$-dimensional quotient linear space
\[\conf \coloneqq \frac{\{(\vv_0,\dots,\vv_{k-1})\colon \vv_i\in\mathbb{R}^2\}} {\{(\xx,\dots,\xx)\colon \xx\in\mathbb{R}^2\}}. \]
Any element of $\conf$ will be called a placement. We denote $(\vv_0;\dots;\vv_{k-1})\in\conf$ as a placement that corresponds to $(\vv_0,\dots,\vv_{k-1})\in(\mathbb{R}^2)^k$.

We define the overlap area function $\area\colon\conf\rightarrow[0,\infty)$ as
\[\area(\vv_0;\dots;\vv_{k-1}) \coloneqq \left| (P_0+\vv_0) \cap \dots \cap (P_{k-1}+\vv_{k-1}) \right|.\]
and then its support $\supp\area$ is compact. To compute $\area(\vv_0;\dots;\vv_{k-1})$ in linear time, we use the following theorem:
\begin{thm}[Shamos]\label{thm:polygon intersection}
  Let $P$ and $Q$ be convex polygons of $m$ vertices and $n$ vertices, respectively. Then $P\cap Q$ can be computed in $O(m+n)$ time.
\end{thm}
\begin{proof}
  This was first proved by Shamos \cite[Section 5.2]{shamos1978computational}; see also \cite[Section 7.6]{o1998computational}.
\end{proof}

The vertices $(x_0,y_0),\dots,(x_{r-1},y_{r-1})$ of the overlap $I$ can be expressed as linear polynomials in $\vv_0,\dots,\vv_{k-1}$ in a generic setting. Ordering them in counter-clockwise direction, the area of $I$ can be computed using the shoelace formula:
\[ |I| = \frac{1}{2}\sum_{i \in \mathbb{Z}/r\mathbb{Z}}(x_iy_{i+1}-x_{i+1}y_i), \]
where the indices are taken modulo $r$. Therefore, $\area$ is a piecewise quadratic function of $\vv_0,\dots,\vv_{k-1}$.

Note that $\area$ may not be quadratic in two cases:
\begin{enumerate}[label=(\Roman*)]
\item an edge of a polygon $P_i+\vv_i$ contains a vertex of another polygon $P_j+\vv_j$ and \label{event:Type I}
\item edges of three distinct polygons $P_i+\vv_i$, $P_j+\vv_j$ and $P_k+\vv_k$ intersect at one point. \label{event:Type II}
\end{enumerate}
Each of these events defines a polytope in $\conf$ of codimension 1. Following \cite{de1998computing}, we call such a polytope as an event polytope. An event polytope defined by \ref{event:Type I} (resp. \ref{event:Type II}) is called of type I (resp. of type II). A hyperplane containing a type I (resp. type II) event polytope is also called of type I (resp. of type II). There are $O(n^2)$ type I hyperplanes and $O(n^3)$ type II hyperplanes.

\section{Linear Programming}
Let $L \subset \conf$ be an $\varepsilon(s)$-translated $m$-flat. The goal of this section is to provide an $O(n)$ time algorithm that finds a placement $\vv\in L$ such that
\[ \area(\vv)\neq 0.\]
If no such placement exists, the algorithm returns $\none$. 

When working with two polygons, $\supp\area$ is simply the Minkowski sum $P_0 + (-P_1)$, where $-P_1$ is the polygon $P_1$ reflected about the origin. However, when working with more than two polygons, the problem becomes more complex. To tackle this problem, we use linear programming with Meggido's solver.
\begin{thm}[Megiddo \cite{megiddo1984}]\label{thm:linear prog}
  A linear programming problem with a fixed number of variables and $n$ constraints can be solved in $O(n)$ time.
\end{thm}

Let $n_i$ be the number of vertices of $P_i$. Then $P_i$ is defined by $n_i$ linear inequalities:
\[f_{i,a}(\xx) \geq 0 \quad\text{(for } a<n_i). \]
The codimension of the $m$-flat $L\subset\conf$ is $2k-m-2$. Thus, $L$ is defined by $\varepsilon(s)$-translated $2k-m-2$ linear equations:
\[g_b(\vv) = 0 \quad\text{(for } b<2k-m-2). \]
Then a point $\xx\in\mathbb{R}^2$ and a placement $\vv = (\vv_0;\dots;\vv_{k-1}) \in \conf$ satisfy the constraints
\begin{equation}\label{eqn:constraints}
  \left\{
    \begin{aligned}
      f_{i,a}(\xx-\vv_i) &\geq 0 \quad\text{(for } i<k \text{ and } a<n_i) \text{ and}\\
      g_b(\vv)         &= 0 \quad\text{(for } b<2k-m-2).
    \end{aligned}
  \right.
\end{equation}
if and only if $\xx\in (P_0+\vv_0) \cap\dots\cap (P_{k-1}+\vv_{k-1})$ and $\vv\in L$. Therefore, we obtain the lemma below.

\begin{lem}
  We have $\vv\in L\cap\supp\area$ if and only if $(\xx,\vv)$ satisfies (\ref{eqn:constraints}) for some $\varepsilon(s)$-translated point $\xx$ in a plane.
\end{lem}

Hence, in $O(n)$ time, we can get $\vv\in L\cap\supp\area$, by solving any linear programming with the constraints (\ref{eqn:constraints}). One problem is that $\vv$ might be on the (topological) boundary of $\supp\area$. 

\begin{lem}\label{lem:interior point}
  Let $M$ be the solution set of $\varepsilon(s)$-translated linear constraints
  \begin{equation}\label{eqn:general constraints}
    \left\{
      \begin{aligned}
        p_i(\xx) &\geq 0 \quad\text{(for $i < n$) and}\\
        q_j(\xx) &=    0 \quad\text{(for $j < m$)}
      \end{aligned}
    \right.
  \end{equation}
  where $\xx \in \mathbb{R}^d$ and $d$ is constant. Then we can compute the maximal affinely independent set $S$ in $O(m+n)$ time.
\end{lem}
\begin{proof}
  By \Cref{thm:linear prog}, we can assume that $M\neq\emptyset$. Moreover, by eliminating variables, we may also assume that $m = 0$. To compute the maximal affinely independent set, we start with an empty set $S$ and gradually add points to it. At each step, we look for a new point that is not in the affine hull of the current set $S$.
  
  To do this, we first select a linear functional $h$ that is non-zero but evaluates to zero on all points in $S$. We can find such a functional in constant time since $d$ is a constant. We then find the minimum and maximum values of $h$ subject to the constraints in $M$, denoted by $\xx_{\min}$ and $\xx_{\max}$, respectively.
  
  If $|S| \leq \dim M$, then $h(\xx_{\min})<h(\xx_{\max})$. Therefore, for some $\xx \in \{\xx_{\min},\xx_{\max}\}$, the set $S\cup\{\xx\}$ should be also affinely independent. In this case, we replace $S$ by $S\cup\{\xx\}$. If not, we terminate the process.
\end{proof}

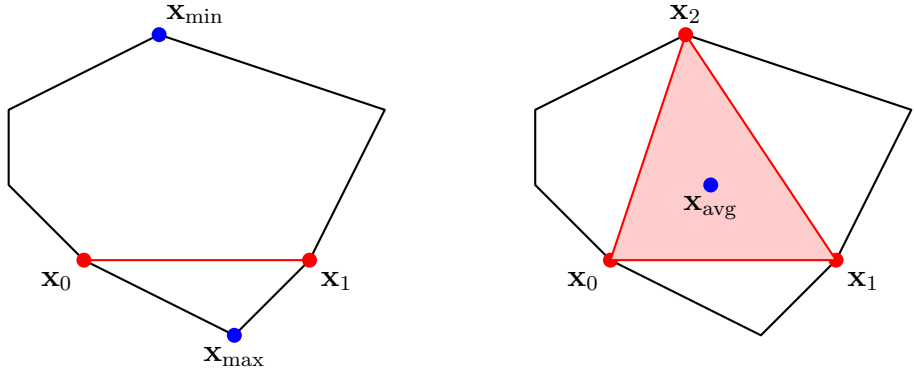
\begin{figure}[ht]
  \centering 
  \newcommand{\findingInteriorPointFigure}[2][]{
    \begin{scope}[#1]
      \coordinate (v0) at (1,1);
      \coordinate (v1) at (3,0);
      \coordinate (v2) at (4,1);
      \coordinate (v3) at (5,3);
      \coordinate (v4) at (2,4);
      \coordinate (v5) at (0,3);
      \coordinate (v6) at (0,2);
      \draw (v0) -- (v1) -- (v2) -- (v3) -- (v4) -- (v5) -- (v6) -- cycle;
      \draw (v0) node[circle,fill,inner sep=2pt,red ]{};
      \draw (v0) node[anchor = north east]{$\xx_0$};
      \draw (v2) node[circle,fill,inner sep=2pt,red ]{};
      \draw (v2) node[anchor = north west]{$\xx_1$};
      #2
    \end{scope}
    }
  \begin{tikzpicture}[thick,line join=round]
    \findingInteriorPointFigure {
      \draw[red] (v0) -- (v2);
      
      \draw (v4) node[circle,fill,inner sep=2pt,blue]{};
      \draw (v4) node[anchor = 210]{$\xx_{\min}$};
      \draw (v1) node[circle,fill,inner sep=2pt,blue]{};
      \draw (v1) node[anchor = north]{$\xx_{\max}$};
      }
     \findingInteriorPointFigure[shift={(7,0)}] {
      \draw[red,fill=red!20] (v0) -- (v2) -- (v4) -- cycle;
      
      \draw (v4) node[circle,fill,inner sep=2pt,red]{};
      \draw (v4) node[anchor = south]{$\xx_2$};
      
      \draw ($1/3*(v0)+1/3*(v2)+1/3*(v4)$) node[circle,fill,inner sep=2pt,blue]{};
      \draw ($1/3*(v0)+1/3*(v2)+1/3*(v4)$) node[anchor = north]{$\xx_{\mathrm{avg}}$}; 
      }
  \end{tikzpicture}
  \caption{Finding a maximal affinely independent set and the topological interior points.}
\end{figure}


  

\begin{thm}\label{thm:nonzero point}
  In $O(n)$ time, we can either return $\vv \in L$ such that $\area(\vv)\neq 0$, or return $\none$ if none exists.
\end{thm}
\begin{proof}
  Let $M \subset \mathbb{R}^2\times L$ be the solution set of the constrains (\ref{eqn:constraints}). Then $\area(\vv) \neq 0$, if and only if $(\xx,\vv)$ is an topological interior point of $M \subset \mathbb{R}^2\times L$ for some $\xx\in\mathbb{R}^2$. Applying \Cref{lem:interior point}, we get the maximal affinely independent set $S$ of $M$.
  
  If $|S| \leq m+2$, then $\dim M < 2+\dim L$, and $M$ has no topological interior point, so we return $\none$. If $|S| = m+3$, then
  \[ (\xx_{\mathrm{avg}},\vv_{\mathrm{avg}}) = \frac{1}{|S|}\sum_{(\xx,\vv)\in S} (\xx,\vv)\]
  is an topological interior points of $M \subset \mathbb{R}^2\times L$. Hence, we return $\vv_{\mathrm{avg}}$.
\end{proof}

\section{Decision Problem}
We aim to find the maximum of $\area$ on an $m$-flat $L\subset\conf$ using an induction on $m$. To do so, we apply a prune-and-search technique on the set of event polytopes. However, this technique requires solving a decision problem: given a hyperplane $H\subset L$, we must determine on which side of $H$ the maximum of $\area|_L$ lies. In this section, we provide an algorithm for this decision problem under certain induction hypotheses.

\begin{thm}\label{thm:log-concavity of area}
    The square root of $\area\colon\conf\rightarrow[0,\infty)$ is concave on its support.
\end{thm}
\begin{proof}
  This follows immediately from the Brunn--Minkowski inequality \cite{minkowski1897}\cite{brunn1887ovale}; see also \cite[Theorem 3.3]{fukuda2007polynomial}.
\end{proof}
Now, we assume the following hypothesis in the rest of this section.

\begin{hyp}\label{hyp:computability of maximum}
  Let $s$ be any constant and $L\subset\conf$ be an $\varepsilon(s)$-translated $(m-1)$-flat. Then we can find $\vv \in L$ maximizing $\area|_L$ in $O(T(n))$ time. 
\end{hyp}

We can partition $L$ into $\varepsilon(s)$-translated open polytopes on which $\area$ is quadratic. Therefore, the maximum $\vv\in L$ of $\area|_L$ is an $\varepsilon(s)$-translated placement.




\begin{thm}\label{thm:hyperplane decision}
  Given an $\varepsilon(s)$-translated $m$-flat $L$ and its $\varepsilon(s)$-translated hyperplane $H\subset L$, let $M\subset L$ be the set of maximum points of $\area|_L$. We can determine which side of $H$ contains $M$ in $O(T(n))$ time.
\end{thm}
\begin{proof}
For any $t\in \mathbb{R}\langle\!\langle \varepsilon_0,\dots \rangle\!\rangle$, let
  \[h(t) = \max_{\vv\in t\nn + H} \area(\vv).\]
Let $N \subset \varepsilon(s)$ be the set of all maximum points of $h(x)$. It suffices to decide on which side $N$ lies with respect to 0. By \Cref{thm:log-concavity of area}, the function $h \colon \varepsilon(s) \rightarrow [0,\infty)_{\varepsilon(s)}$ is unimodal.
\begin{figure}[ht]
  \centering 
  \begin{tikzpicture}[line join=round]
    \draw[-latex] (-3,0) -- (3,0) node[right] {$x$};
    \node at (2.25,1.75) {$y = h(x)$};
    
    \draw[thick,color=blue] (-3,0) -- (-2,0) -- (0,2) -- (2,0) -- (3,0);
  \end{tikzpicture}
  \hspace{20pt}
  \begin{tikzpicture}[line join=round]
    \draw[-latex] (-3,0) -- (3,0) node[right] {$x$};
    \node at (2.25,1.75) {$y = h(x)$};
    
    \draw[thick,color=blue]  (-3,0) -- (-2.5,0);
    \draw[thick,color=blue]  plot[smooth,domain=-2.5:-1.5] (\x, {  (\x+2.5)^2});
    \draw[thick,color=blue]  plot[smooth,domain=-1.5:-0.5] (\x, {2-(\x+0.5)^2});
    \draw[thick,color=blue]  (-.5,2) -- (.5,2);
    \draw[thick,color=blue]  plot[smooth,domain= 0.5: 1.5] (\x, {2-(\x-0.5)^2});
    \draw[thick,color=blue]  plot[smooth,domain= 1.5: 2.5] (\x, {  (\x-2.5)^2});
    \draw[thick,color=blue]  (2.5,0) -- (3,0);
  \end{tikzpicture}
  \caption{Two possible examples of the graph of $h$}\label{fig:unimodal example}
\end{figure}
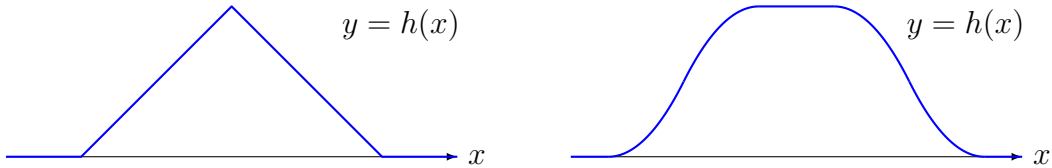

By \Cref{hyp:computability of maximum} with $s+1$, we can compute the sequence
\[ S = \big(h(-\varepsilon_{s+1}), h(0), h(\varepsilon_{s+1})\big)\]
in $O(T(n))$ time. If $h(0) = 0$, then all interior points of $\supp h$ lie in the same side with respect to 0. In this case, apply \Cref{thm:nonzero point} and attempt to get one point of $\supp h$. If $h(0) \neq 0$, there are three remaining cases.

  \begin{enumerate}
  \item If $S$ is strictly increasing, then $N \subset (0,\infty)$.
  \item If $S$ is strictly decreasing, then $N \subset (-\infty,0)$.
  \item If $S$ is not strictly monotonic, then $0 \in N$. \qedhere
  \end{enumerate}
\end{proof}

This proof highlights the necessity of infinitesimal translations for our algorithm. Since $s$ only increases in this step, it is bounded by $\dim \conf = 2k-2$ throughout the paper.

\section{Two Polygons}
The goal of this section is to present a linearithmic 
time algorithm for finding a translation that maximizes the overlap area of two convex polygons under translations. This problem was previously studied by de Berg et al. \cite[Theorem 3.8]{de1998computing}, but our approach is different and allows for handling multiple polygons.

In this section, we only have two convex polygons $P = P_0$ and $Q = P_1$ with $n$ and $m$ vertices, respectively. We consider only one translation vector $\vv = \vv_1 - \vv_0$, and since $\conf$ is two-dimensional, we refer to event polytopes and hyperplanes as event line segments and lines, respectively. Since there are no type II line segments, all event line segments can be defined by one of the following two events:
\begin{enumerate}
\item an edge of a polygon $P$ contains a vertex of polygon $Q+\vv$ and \label{event:Type 01}
\item an edge of a polygon $Q+\vv$ contains a vertex of polygon $P$. \label{event:Type 10}
\end{enumerate}
The first type of event lines segment will be called of type $(0,1)$ and the second type of event lines will be called type $(1,0)$ line segments. The same rules apply to event lines.

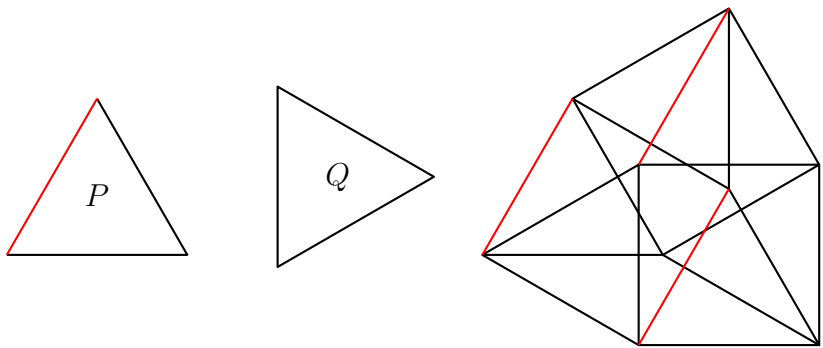
\begin{figure}[ht]
\newcommand*{\eventLineSegment}[3]{%
  \begin{scope}[shift={(-7,-1)}]
    \draw (0,0) \foreach \vPx/\vPy in {#2}{-- (\vPx,\vPy)};
  \draw[red] (0,0) -- (#3);
  \end{scope}
  
  \begin{scope}[shift={(-4,-1.1340)}]
    \draw (0,0) \foreach \vPx/\vPy in {#1}{-- (\vPx,\vPy)} -- cycle;
  \end{scope}
  
  \foreach \vPx/\vPy in {0/0,#2}{
    \begin{scope}[shift={(\vPx,\vPy)}]
      \draw (0,0) \foreach \vQx/\vQy in {#1}{--(-\vQx,-\vQy)} -- cycle;
    \end{scope}}

  \foreach \vQx/\vQy in {0/0,#1}{
    \begin{scope}[shift={(-\vQx,-\vQy)}]
      \draw (0,0) \foreach \vPx/\vPy in {#2}{-- (\vPx,\vPy)};
    \end{scope}}

  \foreach \vQx/\vQy in {0/0,#1}{
    \begin{scope}[shift={(-\vQx,-\vQy)}]
      \draw[red] (0,0) -- (#3);
    \end{scope}}
}
  \centering 
  \begin{tikzpicture}[thick,line join=round,scale=1.2]
    \node at (-6,-1/3){$P$};
    \node at (-4+2/3,-0.1340){$Q$};
    \eventLineSegment{1.7321/1,0/2}{2/0,1/1.7321}{1,1.7321}
  \end{tikzpicture}
  \caption{Event line segments. The parallel lines of one group are highlighted in red.}\label{fig:event line segments}
\end{figure}

Type $(0,1)$ lines are organized into $n$ groups, each with $m$ parallel lines. Our goal is to efficiently prune this set, requiring an appropriate representation. We use 'arrays' to denote sequential data structures with constant time random access, and assume the size of each array is predetermined.

The $n$ groups of parallel lines are represented by sorted arrays $A_0, A_1, \dots, A_{n-1}$. Each array $A_i$ holds the $y$-intercepts and a single slope value for the lines in the $i$-th group. For vertical lines in $A_i$, we store the $x$-intercepts instead.

\begin{defi}
A slope-intercept array $A$ consists of sorted arrays $A_0, A_1, \dots, A_{n-1}$, each with an associated potentially infinite number. Its number of groups is $n$, and its size $|A|$ is the sum of the sizes of $A_i$. Another slope-intercept array $A'$ is a pruned array of $A$ if it consists of $A$ with identical slopes.
\end{defi}

We can use \cite[Theorem 1.4]{firstPaper} to prune a slope-intercept array $A$, but the description is complicated and the result is weaker. Instead, we rely on a stronger version, which we prove in the appendix.

    
    

\begin{thm}\label{thm:prune-and-search}
For a slope-intercept array $A$ with $n$ groups of lines, we can partition the plane $\mathbb{R}^2$ into four closed quadrants $T_0,\dots,T_3$ using one horizontal line $\ell_0$ and one non-horizontal line $\ell_1$. Additionally, for each $i<4$, we can compute pruned array $P_i$ of $A$ that include all lines intersecting the interior of $P_i$ and have size at least $(7/8)|A|$, all in $O(n)$ time.
\end{thm}

Now, we will represent the set of type $(0,1)$ event lines using a slope-intercept array.

\begin{lem}\label{lem:rotating-calipers}
We have $n$ linear functions $f_0,\dots,f_{m-1}$ and $m$ vertices $v_0,\dots,v_{n-1}$ of a convex polygon, both ordered counterclockwise by their gradient vectors and arrangement, respectively. In $O(m+n)$ time, we can find indices $a(0),\dots,a(n-1)$ such that vertex $v_{a(i)}$ minimizes $f_i(v_j)$ for all $j<m$.
\end{lem}
\begin{proof}
In $O(m)$ time, we can find $a(0)$ by computing all $f_0(v_j)$. Now, suppose that $a(i-1)$ is computed. Then compute the sequence $f_i(v_{a(i-1)}),f_i(v_{a(i-1)+1}),f_i(v_{a(i-1)+2}),\dots$ until it increases after some index $a'$. Then $f_i(v_{a'})$ maximizes $f_i$, so $a(i) = a'$. By repeating this process, we can find all $a(0),a(1),\dots,a(m-1)$. Observe that $v_{a(0)}, v_{a(1)}, \dots, v_{a(n-1)}$ are sorted counterclockwise. Since we only perform one rotation, this process requires $O(m+n)$ time.
\end{proof}



\begin{lem}\label{lem:constructing array}
  In $O(m+n)$ time, we can construct a slope-intercept array of $2n$ groups of size $mn$ representing the set of all type $(0,1)$ lines
\end{lem}
\begin{proof}
  Let $P$ be a polygon with $n$ linear inequalities $f_i(\xx) \geq 0$, sorted counterclockwise by the gradients of $\nabla f_i$. Let $\ell_i$ be the line defined by $f_i=0$, and let $v_0,\dots,v_{m-1}$ be the vertices of $Q$ sorted counterclockwise and indexed modulo $m$. Then the set of all type $(0,1)$ lines is 
  \[S = \{-v_j + \ell_i \mid i<n \text{ and } j< m \}. \]
  
By using \Cref{lem:rotating-calipers}, we can determine the indices $a(i)$ and $b(i)$ for each $i$, such that $v_{a(i)}$ (resp. $v_{b(i)}$) is the vertex of $Q$ that minimizes (resp. maximizes) $f_i(v_j)$ for all $j<m$. This computation can be done in $O(m+n)$ time. We can then construct two arrays:
\begin{align*}
A_{2i} &\coloneqq (-v_{a(i)} + \ell_i , -v_{a(i)+1} + \ell_i , \dots, -v_{b(i)-1} + \ell_i) \text{ and}\\
A_{2i+1} &\coloneqq (-v_{b(i)} + \ell_i , -v_{b(i)+1} + \ell_i , \dots, -v_{a(i)-1} + \ell_i),
\end{align*}
whose intercepts are sorted. Note that we do not need to compute the entries of $A_i$ explicitly; once we have computed $a(i)$ and $b(i)$, we can perform random access in $O(1)$ time using the formulas above. The resulting arrays $A_0,\dots,A_{2n-1}$ provide a slope-intercept array representing the set of all type $(0,1)$ lines.
\end{proof}
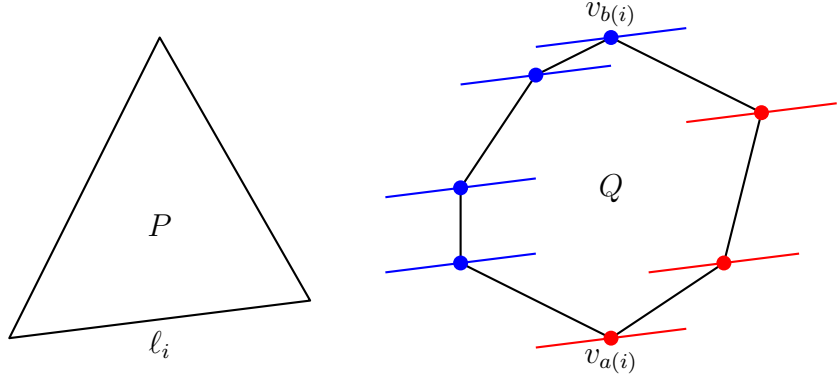
\begin{figure}[ht]
    \centering 
    \begin{tikzpicture}[thick,line join=round]
      \draw (-4,4) -- (-6,0) -- (-2,0.5) node[midway,below]{$\ell_i$} -- cycle;
      \node at (2-6,1.5) {$P$};
      
      \coordinate (v0) at (2  ,0  );
      \coordinate (v1) at (3.5,1  );
      \coordinate (v2) at (4  ,3  );
      \coordinate (v3) at (2  ,4  );
      \coordinate (v4) at (1  ,3.5);
      \coordinate (v5) at (0  ,2  );
      \coordinate (v6) at (0  ,1  );
      
      \draw (v0) node[circle,fill,inner sep=2pt,red ]{} -- 
      (v1) node[circle,fill,inner sep=2pt,red ]{} -- 
      (v2) node[circle,fill,inner sep=2pt,red ]{} -- 
      (v3) node[circle,fill,inner sep=2pt,blue]{} -- 
      (v4) node[circle,fill,inner sep=2pt,blue]{} -- 
      (v5) node[circle,fill,inner sep=2pt,blue]{} -- 
      (v6) node[circle,fill,inner sep=2pt,blue]{} -- 
      cycle;
      
      \draw[red ] (v0)+(-1,-1/8) -- +(1,1/8);
      \draw[red ] (v1)+(-1,-1/8) -- +(1,1/8);
      \draw[red ] (v2)+(-1,-1/8) -- +(1,1/8);
      \draw[blue] (v3)+(-1,-1/8) -- +(1,1/8);
      \draw[blue] (v4)+(-1,-1/8) -- +(1,1/8);
      \draw[blue] (v5)+(-1,-1/8) -- +(1,1/8);
      \draw[blue] (v6)+(-1,-1/8) -- +(1,1/8);
      
      \node[below] at (2,0) {$v_{a(i)}$};
      \node[above] at (2,4) {$v_{b(i)}$};
      \node at (2,2) {$Q$};
    \end{tikzpicture}
    \caption{Visualization of why $A_{2i}$ and $A_{2i+1}$ are sorted.}
  \end{figure}

\begin{thm}\label{thm:baby main theorem}
  Let $P$ and $Q$ be convex polygons, with $m$ and $n$ vertices, respectively. In $O((m+n)\log(m+n))$ time, we can finds a translation $\vv \in \mathbb{R}^2$ maximizing the overlap area
  \[\area(\vv)  = |P \cap (Q+\vv)|.\]
\end{thm}
\begin{proof}
  For any line $\ell \subset \mathbb{R}^2$, we can compute a point $\vv \in \ell$ maximizing $\area|_\ell$ in $O(m+n)$ time by \cite[Corollary 4.1]{avis1996}. Using \Cref{thm:hyperplane decision}, we can determine on which side of $\ell$ the set of maxima of $\area$ lies in $O(m+n)$ time.

By constructing a slope-intercept array $A$ of $(m+n)$ groups with \Cref{lem:constructing array}, we can represent all event lines in $O(m+n)$ time. Applying \Cref{thm:hyperplane decision} to $\ell_0$ and $\ell_1$ obtained from \Cref{thm:prune-and-search}, we can prune $A$ to about 1/8 of its size, and this step requires $O(m+n)$ time. After $O(\log(m+n))$ steps, only $O(1)$ lines remain, and we can find a placement $\vv$ that maximizes the overlap area $\area(\vv)$ directly.
\end{proof}

\section{Several Polygons}

The aim of the section is to give an $O(n\log^{2k-3} n)$ time algorithm to compute $\vv\in\conf$ maximizing $\area$. We first restrict the domain of $\area$ into an $m$-flat $L\subset\conf$ and prove a slightly stronger statement below by induction on $m$.
\begin{thm}\label{thm:stronger main theorem}
  Let $L\subset\conf$ be an $\varepsilon(s)$-translated $m$-flat. Then in $O(n\log^{m-1} n)$ time, we can find $\vv \in L$ maximizing $\area|_L$.
\end{thm}

The proof of the base case can be obtained by modifying the proof of \cite[Corollary 4.1]{avis1996}.

\begin{lem}\label{thm:base case}
  Let $\ell\subset\conf$ be an $\varepsilon(s)$-translated line. Then in $O(n)$ time, we can find $\vv\in\ell$ maximizing $\area|_\ell$.
\end{lem}
\begin{proof}
We parameterize $\ell$ by 
\[f(t) = (f_0(t), f_1(t), \dots, f_{k-1}(t)),\]
where $f_i\colon \mathbb{R}\rightarrow \mathbb{R}^2$ are $\varepsilon(s)$-translated linear functions. We define cylinders 
\[C_i \coloneqq {(x,y,z)\in\mathbb{R}^3,|, (x,y) \in f_i(z)+P_i}.\]
  
  \begin{figure}[ht]
    \centering 
    \tdplotsetmaincoords{110}{0}
    \begin{tikzpicture}[tdplot_main_coords,thick,scale=1.75,line join=round]
      \tikzstyle{facefill} = [fill=magenta]
      \begin{scope}
        \coordinate (A1) at (0  ,0,1);
        \coordinate (B1) at (0.5,2,1);
        \coordinate (C1) at (2  ,0,1);
        
        \coordinate (D0) at (0.75-0.25*1.5,1,-0.5);
        \coordinate (D1) at (0.75     ,1,1);
        \coordinate (D2) at (0.75+0.25*1.5,1, 2.5);
        
        \draw (D0) -- (D1);
        \draw[facefill] (A1) -- (B1) -- (C1) -- cycle;
        \draw (D1) -- (D2);
        
        \node at (0,1,1) {$P_j$};
        \node at (1.6,1,2.25) {$(f_j(t),t)$};
      \end{scope}
      \begin{scope}[shift={(4,0)}]
        \coordinate (A0) at ( 0  ,0,2);
        \coordinate (B0) at ( 0.5,2,2);
        \coordinate (C0) at ( 2  ,0,2);
        \coordinate (A1) at (-0.5,0,0);
        \coordinate (B1) at ( 0  ,2,0);
        \coordinate (C1) at ( 1.5,0,0);
        
        \coordinate (A2) at ( 0.0625,0, 2.25);
        \coordinate (B2) at ( 0.5625,2, 2.25);
        \coordinate (C2) at ( 2.0625,0, 2.25);
        \coordinate (A3) at (-0.5625,0,-0.25);
        \coordinate (B3) at (-0.0625,2,-0.25);
        \coordinate (C3) at ( 1.4375,0,-0.25);
        
        \fill[facefill,fill opacity=0.5] (C2) -- (C3) -- (A3) -- (A2) -- cycle;
        \fill[facefill] (A2) -- (A3) -- (B3) -- (B2) -- cycle;
        \fill[facefill] (B2) -- (B3) -- (C3) -- (C2) -- cycle;
        \draw[dotted] (A0) -- (B0) -- (C0) -- cycle;
        \draw[dotted] (A1) -- (B1) -- (C1) -- cycle;
        \draw         (A2) -- (A3);
        \draw         (B2) -- (B3);
        \draw         (C2) -- (C3);
        
        \node at (-0.66,1,1) {$C_j$};
      \end{scope}
    \end{tikzpicture}
    \caption{Depicting the cylinder $C_i$ obtained from $P_i$ and $\ell$.}\label{fig:cylinder}
  \end{figure}
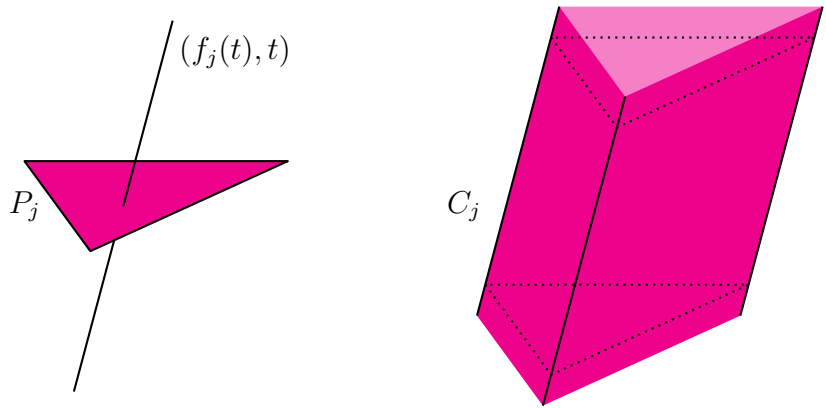

We can compute $C = C_0\cap C_1\cap\dots\cap C_{k-1}$ in $O(n)$ time using Chazelle's algorithm \cite{chazelle1993b}. Let $H_t\subset\mathbb{R}^3$ be the hyperplane defined by $z = t$. Then we have $|C\cap H_t| = |(P_0+f_0(t))\cap\dots\cap(P_{k-1}+f_{k-1}(t))|$. We can find $t$ maximizing $|C\cap H_t|$ in $O(n)$ time using \cite[Theorem 3.2]{avis1996}. For such a $t$, the maximum point of $\area|_\ell$ is $f(t)\in\ell$.
\end{proof}

Therefore, we assume that $m>1$ and the following induction hypothesis is true.

\begin{hyp}\label{hyp:induction}
  Let $L\subset\conf$ be an $\varepsilon(s)$-translated $(m-1)$-flat. Then we can find $\vv \in L$ maximizing $\area|_L$ in $O(n\log^{m-2} n)$ time.
\end{hyp}

We will first find an $m$-simplex $T_I\subset L$ such that $T_I$ has the maximum point of $\Pi|L$ and no type I hyperplane intersects the interior of $T_I$. Recall that type I hyperplanes are defined by the following event.
\begin{itemize}
\item[\ref{event:Type I}] an edge of a polygon $P_i+\vv_i$ contains a vertex of another polygon $P_j+\vv_j$ and 
\end{itemize}
If $i$ and $j$ are specified, then it will be called a type $(i,j)$ hyperplane. Then type I hyperplanes are grouped into $k(k-1)$ groups, each of which is the set of type $(i,j)$ hyperplanes. Any type $(i,j)$ hyperplane $H$ is defined by a linear equation of the form
\[\nn\cdot(\xx_i-\xx_j) = c\]
for some $\nn\in\mathbb{R}^2$ and $c\in\mathbb{R}$. Consider the projection
\begin{align*}
  \pi_{i,j}\colon \conf &\rightarrow \mathbb{R}^2\\
  \xx &\mapsto \xx_i - \xx_j.
\end{align*}
Then $\pi_{i,j}(H)\subset\mathbb{R}^2$ is a line. Such a line will also be called of type $(i,j)$. Thus, we will find a triangle $T_{i,j}\subset L$ such that no type $(i,j)$ lines intersect the interior of $T_{i,j}$.

\begin{prop}\label{prop:prune I ij}
In $O(n\log^{m-1}n)$ time, We can find a triangle $T_{i,j}\subset\mathbb{R}^2$ such that
\begin{enumerate}
  \item a maximum point of $\area|_L$ lies on $\pi_{i,j}^{-1}(T_{i,j}) \cap L$, and
  \item no type $(i,j)$ lines intersects the interior of $T_{i,j}$.
  \end{enumerate}
\end{prop}
\begin{proof}
The proof is similar to that of \Cref{thm:baby main theorem}. Let $M \subset L$ be the set of placements maximizing $\area|_L$. To determine on which side of a line $\ell$ the set $\pi_{i,j}(M)$ lies, we apply \Cref{thm:hyperplane decision}, which takes $O(n\log^{m-2}n)$ time.

We can represent all type-$(i,j)$ lines by a slope-intercept array $A$ in $O(n)$ time, as shown in \Cref{lem:constructing array}. Applying \Cref{thm:prune-and-search} to obtain lines $\ell_0$ and $\ell_1$, we can prune $A$ to about 1/8 of its size using \Cref{thm:hyperplane decision}. This step requires $O(n\log^{m-2}n)$ time. After $O(\log n)$ steps, only $O(1)$ lines remain, and then we triangulate the remaining region. This gives a triangle $T_{i,j}$ with the desired properties in $O(n\log^{m-1}n)$ time.
\end{proof}

Now, define
\begin{equation}\label{def:prune I}
  T_I \coloneqq \bigcap_{i,j < d} \pi_{i,j}^{-1}(T_{i,j})\subset L.
\end{equation}
Then $T_I$ is defined by $3 k(k-1) \in O(1)$ linear polynomials, and by construction, no type I hyperplanes intersect the interior of $T_I$. Our goal now is to find an $m$-simplex $T\subset T_I$ such that $T$ has the maximum point of $\Pi|_L$ and no event polytopes intersect the interior of $T$.

To achieve this, we first note that only $O(n)$ type II hyperplanes intersect the interior of $T_I$. Thus, we can obtain $T$ by repeatedly applying Chazelle's cutting algorithm.

\begin{defi}[Matou\v{s}ek \cite{matouvsek1990cutting}]\label{def:cutting}
  A cutting of $\mathbb{R}^d$ is a collection $C$ of possibly unbounded $d$-simplices with disjoint interiors, which together cover $\mathbb{R}^d$. Let $S$ be a set of $n$ hyperplanes in $\mathbb{R}^d$. Then a cutting $C$ is a $(1/2)$-cutting for $S$ if the interior of each simplex intersects at most $n/2$ hyperplanes. 
\end{defi}

\begin{thm}[Chazelle \cite{chazelle1993}]\label{thm:cutting}
  With the notation in \Cref{def:cutting}, a $(1/2)$-cutting of size $O(2^d)$ can be computed in $O(n2^{d-1})$ time. In addition, the set of hyperplanes intersecting each simplex of the cutting is reported in the same time. 
\end{thm}

\begin{prop}\label{prop:prune II}
  In $O(n\log^{m-1}n)$ time, we can find an $\varepsilon(s)$-translated $m$-simplex $T\subset L$ such that
  \begin{enumerate}
  \item the maximum point of $\area|_L$ lies on $T$, and
  \item no event polytope intersects the interior of $T$.
  \end{enumerate}
\end{prop}
\begin{proof}
 Take $T_I$ as defined in (\ref{def:prune I}). By construction, no type I hyperplane intersects the interior of $T_I \subset L$. Therefore, the set of pairs of intersecting edges of $P_i$ and $P_j$ does not depend on the placement $\vv \in T_I$. Moreover, every edge of $P_i$ intersects at most two edges of $P_j$. Therefore, there are at most
  \[\binom{d}{3} 4n \in O(n)\]
type II polytopes intersecting the interior of $T_I$. In $O(n)$ time, we can compute the set $S$ containing all such type II hyperplanes by sampling a placement $\vv$ in the interior of $T_I$.
  
To find a simplex $T$ satisfying the conditions of \Cref{prop:prune II}, we first set $T = T_I$. Then we define $S$ as the set of hyperplanes in $L$ containing a facet of $T$ or a type II polytope that intersects the interior of $T$. We can compute a $(1/2)$-cutting $C$ of size $O(1)$ for $S$ in $O(n)$ time using \Cref{thm:cutting}. Using \Cref{thm:hyperplane decision}, we can then find a simplex $T' \in C$ containing the maximum point of $\area|_L$ in $O(n\log^{m-2}n)$ time. We set $T = T'$ and repeat this process $O(\log n)$ times until no type II polytopes intersect the interior of $T$.
\end{proof}

\begin{proof}[proof of \Cref{thm:stronger main theorem}]
  We can find $T$ as in \Cref{prop:prune II} and compute $\Pi|_T$, which is a quadratic polynomial. Then we can directly compute the maximum point of $\Pi|_T$.
\end{proof}

\mainTheorem*
\begin{proof}
  This is a corollary of \Cref{thm:stronger main theorem} with $R = \mathbb{R}$ and $m = 2k-2$.
\end{proof}

\section{Set of Maxima}
Our next step is to determine the set $M\subset\conf$ of placements $\vv\in\conf$ that maximize the overlap area $\area$. Once we identify at least one such placement, the problem becomes easy, as every maximal overlap is the same up to translation. To accomplish this, we rely on the equality condition of the Brunn-Minkowski inequality.

\begin{thm}[Minkowski]\label{thm:B-M equality}
Let $A$ and $B$ be compact subsets of $\mathbb{R}^2$ with nonzero area. Then 
\[\left|\frac{1}{2}A+\frac{1}{2}B\right|^{1/2} \geq \frac{1}{2}|A|^{1/2}+\frac{1}{2}|B|^{1/2}, \]
and the equality holds if and only if $A$ and $B$ are homothetic.
\end{thm}

We define $I(\vv)$ for any placement $\vv \in \conf$, as follows:
\[I(\vv)\coloneqq (P_0+\vv_0) \cap \dots \cap (P_{k-1}+\vv_{k-1}).\]

\begin{lem}
Let $\vv,\uu\in\conf$ be two placements that both maximize $\area$. Then $I(\uu)$ and $I(\vv)$ are equivalent up to translation.
\end{lem}
\begin{proof}
  Since $P_0,\dots,P_{k-1}$ are convex,
  \[ \frac{1}{2}I(\uu) + \frac{1}{2}I(\vv) \subset I\left(\frac{\uu+\vv}{2}\right). \]
  Therefore,
  \[  \left|\frac{1}{2}I(\uu) + \frac{1}{2}I(\vv)\right| \leq \left|I\left(\frac{\uu+\vv}{2}\right)\right| \leq |I(\vv)|. \]
  As a result, $I(\uu)$ and $I(\vv)$ are homothetic by \Cref{thm:B-M equality}. Since $|I(\vv)| = |I(\uu)|$, this implies that $I(\uu)$ and $I(\vv)$ are equivalent up to translation.
\end{proof}

  We then fix a maximal overlap $I_{\max} \subset \mathbb{R}^2$. The set of all $\vv_i$ such that $I_{\max} \subset \vv_i + P_i$ is given by the Minkowski difference
\begin{align*}
(-P_i) - (-I_{\max}) &= \{\xx \in \mathbb{R}^2 \mid \xx + (-I_{\max}) \subset - P_i\} \\
&= \{\xx \in \mathbb{R}^2 \mid I_{\max} \subset \xx + P_i\}.
\end{align*}
We define $N \coloneqq \prod_{i<m} (P_i - I_{\max})$ and let $\pi\colon(\mathbb{R}^2)^k \rightarrow\conf$ be the natural quotient. 
\begin{lem}
  The restricted map $\pi|_N\colon N \rightarrow M$ is an affine isomorphism.
\end{lem}
\begin{proof}
  By construction $M = \pi(N)$. Suppose there exist two distinct $\uu, \vv \in N$ such that
  \[\uu = \vv + (\xx, \xx, \dots, \xx)\]
  for some $\xx \in \mathbb{R}^2$. This implies that $I_{\max} = I(\vv)$ and $I_{\max} = I(\uu) = I(\vv) + \xx$. As a result, we must have $\uu = \vv$.
\end{proof}

Since each $P_i$ and $I_{\max}$ contain at most $n$ vertices, we can represent $(-P_i) - (-I_{\max})$ using $O(n)$ linear constraints without redundancy. This computation can be completed in $O(n)$ time. Consequently, by employing standard linear algebra techniques, we can describe $M \subset \conf$ using $O(n)$ linear constraints without redundancy in $O(n)$ time.

\begin{thm}\label{thm:better second main}
In $O(n)$ time, we can represent $M \subset \conf$ using $O(n)$ linear constraints without redundancy.
\end{thm}
\begin{proof}
  Let $\vv_i = (x_i, y_i)$ for each $i < m$. A linear polynomial $f(\vv_0, \dots, \vv_{k-1})$ can be written as an affine combination of $\vv_1 - \vv_0, \dots, \vv_{k-1} - \vv_0$ if and only if
  \[ \sum_{i<m}\frac{\partial}{\partial x_i} f = 0 \quad\text{and}\quad \sum_{i<m}\frac{\partial}{\partial y_i} f = 0.\]
Every edge of $I_{\max}$ should be part of an edge of $P_i$ for some $i < m$. Consider two nonparallel edges. They yield two linear equations:
    \[\ab\cdot\vv_i = c \quad\text{and}\quad \bb\cdot\vv_j - d.\]
    Here, $\vv_i$ and $\vv_j$ are column vectors, and $\ab$ and $\bb$ are row vectors. Let
    \[\vv' = 
    \begin{pmatrix}
        x'\\
        y'
    \end{pmatrix}
    \coloneqq
    {\begin{pmatrix}
        \ab\\
        \bb
    \end{pmatrix}}^{-1}
    \begin{pmatrix}
        \ab\cdot\vv_i - c\\
        \bb\cdot\vv_j - d
    \end{pmatrix}.\]
    Then 
    \[ \sum_{i<m}\frac{\partial}{\partial x_i} \vv' = 
    \begin{pmatrix}
        1\\
        0
    \end{pmatrix} 
    \quad\text{and}\quad 
    \sum_{i<m}\frac{\partial}{\partial y_i} \vv' = 
    \begin{pmatrix}
        0\\
        1
    \end{pmatrix},\]
    so we replace every $\vv_i$ by $\vv_i - \vv'$ in the linear constraints. As a result, each constraint is expressed in terms of $\vv_1 - \vv_0, \dots, \vv_{k-1} - \vv_0$.
    \end{proof} 

\Cref{thm:second main} is an immediate corollary of \Cref{thm:better second main}.

\appendix
\section{Partitioning with Two Lines}

In this section, we prove \Cref{thm:prune-and-search}. While the main theorems can be derived solely from \cite[Theorem 1.4]{firstPaper}, this approach is somewhat unsatisfactory. Specifically, it requires three queries at every step and prunes only $1/18$ of the lines, leading to a slowdown factor of $27/8$. Moreover, the statement of \cite[Theorem 1.4]{firstPaper} is much more difficult to describe.

To provide a more convenient (at least in the authors' taste) proof, we instead prove the dual statement. This is the problem of partitioning a set of points in the plane with two lines such that each quadrant contains at least $1/8$ of the points. We begin by presenting Megiddo's linear time algorithm for a special case of the ham sandwich problem \cite[Section 2]{megiddo1985}.

\begin{thm}\label{thm:ham sandwich}
  Given two finite sets of points in the plane with a total of $n$ points, and with disjoint convex hulls, we can compute a line that bisects both sets in $O(n)$ time.
\end{thm}

The following corollary is a slightly stronger result than Megiddo's original main theorem \cite{megiddo1985}.

\begin{cor}\label{thm:point prune}
Given a set of $n$ points in a projective plane $\mathbb{P}^2$, we can compute a horizontal line $\ell_0$ and a non-horizontal line $\ell_1$ in $O(n)$ time, such that each closed quadrant defined by the two lines contains at least $\lfloor n/4 \rfloor$ points in $O(n)$ time.
\end{cor}
\begin{proof}
First, we can assume that there are no points on the line at infinity by applying the perturbation $(a; b; c) \mapsto (a; b; c + \varepsilon b)$. An appropriate value for $\varepsilon$ can be computed in $O(n)$ time. Additionally, we can disregard a single point at $(1; 0; 0)$, as it is contained in all closed quadrants. 

Next, we identify the horizontal line that passes through the median $y$-coordinate of the points, denoted as $\ell_0$. If $\ell_0$ contains at least half of the points, we can select any non-horizontal line $\ell_1$ that passes through the median point $m$ of $\ell_0$. As a result, we assume that $\ell_0$ contains fewer than half of the points.

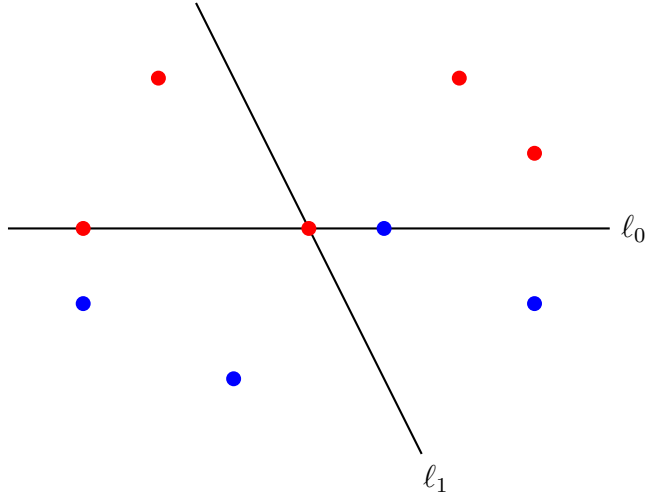
\begin{figure}[ht]
  \centering 
  \begin{tikzpicture}[thick,line join=round]
    \draw (-4  ,0) -- (4  , 0) node[right]     {$\ell_0 $};
    \draw (-1.5,3) -- (1.5,-3) node[anchor=120]{$\ell_1 $};

    \draw (-2, 2) node[circle,fill,inner sep=2pt,red]{};
    \draw ( 2, 2) node[circle,fill,inner sep=2pt,red]{};
    \draw ( 3, 1) node[circle,fill,inner sep=2pt,red]{};
    \draw (-3, 0) node[circle,fill,inner sep=2pt,red]{};
    \draw ( 0, 0) node[circle,fill,inner sep=2pt,red]{};

    \draw ( 1, 0) node[circle,fill,inner sep=2pt,blue]{};
    \draw (-3,-1) node[circle,fill,inner sep=2pt,blue]{};
    \draw ( 3,-1) node[circle,fill,inner sep=2pt,blue]{};
    \draw (-1,-2) node[circle,fill,inner sep=2pt,blue]{};
  \end{tikzpicture}
  \caption{The red represents $A$ and the blue represents $B$}\label{fig:ham sandwich}
\end{figure}

We put the points above the line $\ell_0$ in a set $A$. Moreover, we also put points on $\ell_0$ from left until $A$ has at least half of the points. Then $B$ is the set of remaining points. Since the convex hulls of $A$ and $B$ are disjoint, we can apply \Cref{thm:ham sandwich} to compute the line $\ell_1$ that simultaneously bisects both sets. Since $\ell_0$ contains less than half of the points, $\ell_1$ should not be horizontal. This divides the plane into four closed quadrants, each containing at least $\lfloor n/4 \rfloor$ points.
\end{proof}

An intersecting aspect is that \Cref{thm:point prune} offers a linear-time algorithm for its own weighted version. It is important to note that this approach heavily relies on the following well-established result.
\begin{lem}
    Given $n$ distinct real numbers with positive weights, we can determine the weighted median of these numbers in $O(n)$ time.
\end{lem}

\begin{thm}\label{thm:point weighted prune}
Given $n$ weighted points in a projective plane $\mathbb{P}^2$ with positive weights $\lambda_0,\dots,\lambda_{n-1}$, we can compute a horizontal line $\ell_0$ and a non-horizontal line $\ell_1$ in $O(n)$ time such that each closed quadrant defined by the two lines contains at least $1/4$ of the total weight.
\end{thm}
\begin{proof}
  Once again, we can assume that there are no points on the line at infinity by applying perturbation $(a;b;c) \mapsto (a; b; c+\varepsilon b)$ and ignoring a single point at $(1, 0, 0)$. Let $\ell_0$ be the weighted median horizontal line. If $\ell_0$ contains at least half of the total weight, then we can choose any non-horizontal line $\ell_1$ passing through the weighted median point $m$ of $\ell_0$. Therefore, we assume that $\ell_0$ contains less than half of the total weight.

We start by putting all points above the line $\ell_0$ into a set $A$, and adding points on $\ell_0$ from left to right until $A$ has at least half of the total weight. We modify the weight of the last point $p$ so that the total weight of $A$ is exactly half of the total weight, and set $B$ as the remaining points and $p$ with the remaining weight.

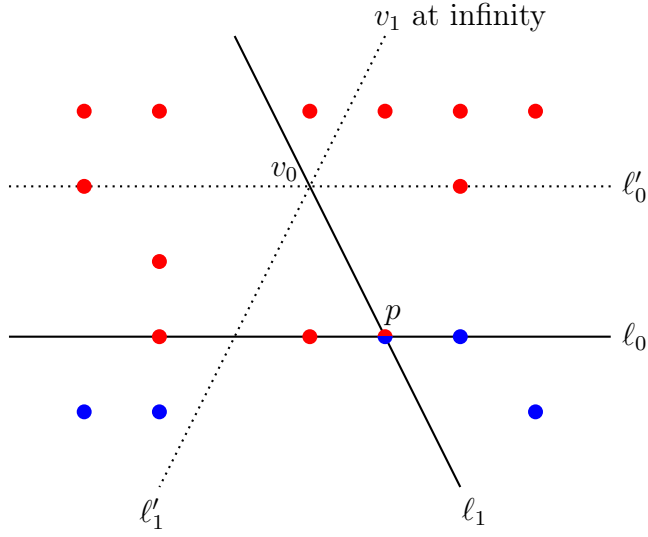
\begin{figure}[ht]
  \centering 
  \begin{tikzpicture}[thick,line join=round]
    \draw[dotted] (-4, 0) -- ( 4  , 0) node[right]     {$\ell'_0$};
    \draw         (-4,-2) -- ( 4  ,-2) node[right]     {$\ell_0 $};
    \draw[dotted] ( 1, 2) -- (-2,-4) node[anchor=75] {$\ell'_1$};
    \draw         (-1, 2) -- ( 2,-4) node[anchor=120]{$\ell_1 $};
    
    \draw (-3, 1) node[circle,fill,inner sep=2pt,red]{};
    \draw (-2, 1) node[circle,fill,inner sep=2pt,red]{};
    \draw ( 0, 1) node[circle,fill,inner sep=2pt,red]{};
    \draw ( 1, 1) node[circle,fill,inner sep=2pt,red]{};
    \draw ( 2, 1) node[circle,fill,inner sep=2pt,red]{};
    \draw ( 3, 1) node[circle,fill,inner sep=2pt,red]{};
    \draw (-3, 0) node[circle,fill,inner sep=2pt,red]{};
    \draw ( 2, 0) node[circle,fill,inner sep=2pt,red]{};
    \draw (-2,-1) node[circle,fill,inner sep=2pt,red]{};
    \draw (-2,-2) node[circle,fill,inner sep=2pt,red]{};
    \draw ( 0,-2) node[circle,fill,inner sep=2pt,red]{};

    \begin{scope}
    \clip(0,-2) rectangle (2,-1);
    \draw (1,-2) node[circle,fill,inner sep=2pt,red]{};
    \end{scope}
    \begin{scope}
    \clip(0,-2) rectangle (2,-3);
    \draw (1,-2) node[circle,fill,inner sep=2pt,blue]{};
    \end{scope}

    \draw ( 2,-2) node[circle,fill,inner sep=2pt,blue]{};
    \draw (-3,-3) node[circle,fill,inner sep=2pt,blue]{};
    \draw (-2,-3) node[circle,fill,inner sep=2pt,blue]{};
    \draw ( 3,-3) node[circle,fill,inner sep=2pt,blue]{};

    \node[anchor=-30] at(0,0){$v_0$};
    \node[anchor=250] at(1,-2){$p$};
    \node             at(2,2.25){$v_1$ at infinity};
    
  \end{tikzpicture}
  \caption{The red represents $A$ and the blue represents $B$}\label{fig:point group ham sandwich}
\end{figure}

Since $\ell_0$ contains less than half of the total weight, any ham sandwich cut of $A$ and $B$ must not be horizontal. We can then find two lines $\ell_0'$ and $\ell_1'$ as in \Cref{thm:ham sandwich}. Let $v_0$ be their intersection, and let $v_1$ be the intersection of $\ell_1'$ and the line at infinity.

Without loss of generality, we may assume that the $y$-coordinate of $\ell_0$ is at most that of $\ell'_0$. We then take a line $\ell_1$ passing through $v_0$ and bisecting the weight of $B$. If $\ell_1$ also bisects the weight of $A$, then this is the desired line. Otherwise, we may assume without loss of generality that the left side of $\ell_1$ contains more weight. Then any ham sandwich cut of $A$ and $B$ must pass through the left side of $\ell_0'$ with respect to $v_0$.

We can repeat this process with $v_1$. Then we determine which side of the line at infinity a ham sandwich cut of $A$ and $B$ must pass through with respect to $v_1$. After this, we identify one quadrant that does not intersect any ham sandwich cut of $A$ and $B$. Thus, we can merge the points in that quadrant into two points, one for $A$ and one for $B$, and repeat the entire process.

  Every step, the number of points become $3/4$ and we get at most $3$ new points. Thus, in $O(n)$ time, at most 12 points remains. Then we can get a ham sandwich cut of $A$ and $B$ by brute force. The ham sandwich theorem implies that such a cut exists.
\end{proof}

\begin{thm}\label{thm:point group prune}
Let $A$ be an array of arrays $A_0,\dots,A_{n-1}$ of points. Suppose that for each $i<m$, points on $A_i$ lie on the same horizontal line and are sorted from left to right. Then, in $O(n)$ time, we can find $\ell_0$ and $\ell_1$ such that for each $i<4$, we can obtain a pruned array $P_i$ of $A$ with $|P_i| \geq |A|/8$ and $P_i$ contained in the $i$th quadrant.
\end{thm}
\begin{proof}
    We can simply choose median points of each of $A_i$, and let the weight be the size of $A_i$. Then we can apply \Cref{thm:point weighted prune} and get the answer.
\end{proof}

Let $(\mathbb{P}^2)^\dual$ be the dual projective space, the space parametrizing lines on $\mathbb{P}^2$. Consider the map
\begin{align*}
(\mathbb{P}^2)^\dual &\rightarrow \mathbb{P}^2\\
ax+by+cz=0 &\mapsto (c;b;a).
\end{align*}
Then \Cref{thm:prune-and-search} is exactly the dual theorem of \Cref{thm:point group prune} under this map. In fact, we can do a little better if extra time is allowed.

\begin{lem}\label{lem:multiple arrays ranking}
  Let $S$ be a collection of $m$ sorted arrays. Given $x$, we can compute the rank of $x$ in $O(n\log|S|)$ time using binary search on each array.
\end{lem}
\begin{proof}
    We can apply binear search on each array and get the answer.
\end{proof}

\begin{lem}\label{lem:multiple arrays median}
    Let $S$ be a collection of $m$ sorted arrays. Then we can find the $i$th element of $S$ in $O(n\log^2|S|)$ time.
\end{lem}
\begin{proof}
Let $m$ be the weighted median of medians of each array, where the weight is given by the size of each array. By applying binary search on each array, we can compute the rank of $m$ in $O(n\log|S|)$ time. From the medians, we can then discard $1/4$ of the elements of $S$ and recursively repeat the process. Since there are $O(\log|S|)$ levels of recursion, the overall time complexity is $O(n\log^2|S|)$.
\end{proof}

\begin{thm}\label{thm:exact point group prune}
Let $A$ be an array of arrays $A_0,\dots,A_{n-1}$ of points. Suppose that for each $i<m$, points on $A_i$ lie on the same horizontal line and are sorted from left to right. Then, in $O(n\log^2|S|)$ time, we can find $\ell_0$ and $\ell_1$ such that for each $i<4$, we can obtain a pruned array $P_i$ of $A$ with $|P_i| \geq |A|/4$ and $P_i$ contained in the $i$th quadrant.
\end{thm}
\begin{proof}
The proof is almost same are that of \Cref{thm:point weighted prune}. However, we need to use \Cref{thm:point group prune} for pruning points, \Cref{lem:multiple arrays median} for bisecting $B$, and \Cref{lem:multiple arrays ranking} for counting points of $A$.
\end{proof}

\bibliographystyle{plain}
\bibliography{ref}

\end{document}